\documentclass[number,3p,twocolumn,preprint]{elsarticle}

\usepackage[active]{srcltx}
\usepackage{epsfig}
\usepackage{amsthm}
\usepackage{amssymb}
\usepackage{hyperref}
\usepackage{xargs}
\usepackage[matrix,frame,arrow]{xy}
\usepackage{amsmath}

\newcommand{\qw}[1][-1]{\ar @{-} [0,#1]}

\newcommand{\measureD}[1]{*{\xy*+=+<.5em>{\vphantom{\rule{0em}{.1em}#1}}*\cir{r_l};p\save*!R{#1} \restore\save+UC;+UC-<.5em,0em>*!R{\hphantom{#1}}+L **\dir{-} \restore\save+DC;+DC-<.5em,0em>*!R{\hphantom{#1}}+L **\dir{-} \restore\POS+UC-<.5em,0em>*!R{\hphantom{#1}}+L;+DC-<.5em,0em>*!R{\hphantom{#1}}+L **\dir{-} \endxy} \qw}
\newcommand{\Qcircuit}[1][0em]{\xymatrix @*=<#1>}

\newcommand{\pureghost}[1]{*+<1em,.9em>{\hphantom{#1}}}
\newcommand{\multiprepareC}[2]{*+<1em,.9em>{\hphantom{#2}}\save[0,0].[#1,0];p\save !C
  *{#2},p+RU+<0em,0em>;+LU+<+.8em,0em> **\dir{-}\restore\save +RD;+RU **\dir{-}\restore\save
  +RD;+LD+<.8em,0em> **\dir{-} \restore\save +LD+<0em,.8em>;+LU-<0em,.8em> **\dir{-} \restore \POS
  !UL*!UL{\cir<.9em>{u_r}};!DL*!DL{\cir<.9em>{l_u}}\restore}
\newcommand{\prepareC}[1]{*{\xy*+=+<.5em>{\vphantom{#1\rule{0em}{.1em}}}*\cir{l^r};p\save*!L{#1} \restore\save+UC;+UC+<.5em,0em>*!L{\hphantom{#1}}+R **\dir{-} \restore\save+DC;+DC+<.5em,0em>*!L{\hphantom{#1}}+R **\dir{-} \restore\POS+UC+<.5em,0em>*!L{\hphantom{#1}}+R;+DC+<.5em,0em>*!L{\hphantom{#1}}+R **\dir{-} \endxy}}

\newtheorem{Def}{Definition}

\newtheorem{corollary}{Corollary}
\newtheorem*{corollary*}{Corollary}
\newtheorem{theorem}{Theorem}

\newcommand{\Ket}[1]{ \left| #1 \right\rangle}

\newcommand{\Bra}[1]{ \left\langle #1 \right|}

\newcommand{\KetBra}[2]{\Ket{#1}\!\!\Bra{#2}}

\newcommand\State[1]{\left|#1\right)}

\newcommand\Effect[1]{\left(#1\right|}

\newcommand{\RBraKet}[2]{(#1|#2)}

\newcommandx\ProbCond[4][usedefault, addprefix=\global, 1=,
2=]{\operatorname{Pr}_{#1}^{#2}\!\left[#3\vphantom{#4} \right| \left.#4\vphantom{#3}\right]}

\newcommand\ProbCondTilde[2]{\widetilde{\operatorname{Pr}}\!\left[#1\vphantom{#2} \right|
\left.#2\vphantom{#1}\right]}

\newcommandx\Prob[2][usedefault, addprefix=\global, 1=]{\operatorname{Pr}_{#1}\!\left[#2\right]}

\newcommand{\ustickcool}[1]{*!D!<0em,-.1em>=<0em>{\scriptstyle #1}}

\begin{document}
\title{Spooky action at a distance in general probabilistic theories}

\author[quit,infn]{Giacomo Mauro D'Ariano\fnref{url}}
\ead{dariano@unipv.it}
\author[quit]{Franco Manessi\fnref{url}}
\ead{franco.manessi01@ateneopv.it}
\author[quit,infn]{Paolo Perinotti\fnref{url}} 
\ead{paolo.perinotti@unipv.it}
\address[quit]{QUIT group,  Dipartimento di Fisica, via Bassi   6, 27100 Pavia, Italy.} 
\address[infn]{INFN Gruppo IV, Sezione di Pavia, via Bassi, 6, 27100 Pavia, Italy.} 
\date{\today}
\fntext[url]{URL: \url{http://www.quantummechanics.it}}

\begin{abstract}
  We call a probabilistic theory ``complete'' if it cannot be further
  refined by no-signaling hidden-variable models, and name a theory
  ``spooky'' if every equivalent hidden-variable model violates
  Shimony's Outcome Independence. We prove that a complete theory is
  spooky if and only if it admits a pure steering state in the sense
  of Schr\"odinger. Finally we show that steering of complementary
  states leads to a Schr\"odinger's cat-like paradox.
\end{abstract}
\begin{keyword}
quantum theory \sep hidden variable theories \sep Schr\"odinger's cat\sep steering
\end{keyword}

\maketitle
\bibliographystyle{elsarticle-num}
\section{Introduction}

Since the early days physicists have been wondering whether Quantum Theory (QT) can be considered
complete \cite{vonNeumann1932, EPR}, or more refined theories compatible with quantum predictions
could exist. These models,
also known as Hidden Variable Theories (HVT), reproduce QT thanks to a
statistical definition of pure quantum states, which are obtained as averages over the more fundamental
states of the HVT. In this approach, which reduces QT to a Statistical Mechanics, many results have
been obtained, such as the theorems by Kochen-Specker and Bell \cite{Bell_1964,
  KSTheorem}, and the results by Conway-Kochen on the free will \cite{FreeWill, StrongFreeWill}.

Recently, General Probabilistic Theories (GPT)  have received great attention as the appropriate framework to study foundational aspects
of physics \cite{QUIT-Arxiv, Hardy:2001jk,
  QUIT-ProbTheories, PhilosophyQuantumInformationEntanglement2, 2009arXiv0911.0695D, Masanes,
  Physics.4.55}. Despite much work has been devoted to the relations between probabilistic theories and
HVTs, these results are mostly a characterization of the probability measures, lacking a conceptual
physical characterization of the theory itself, for example in terms of axioms. So far there exist
examples of probability measures that do not respect locality, signaling, non-contextuality,
determinism, completeness, etc., but none of these highlights the physical properties that a GPT must
fulfill in order to achieve such violations.

The present Letter breaks the ground in the direction of providing a characterization theorem for
complete ``spooky'' theories (see definitions in the following). Roughly speaking, the spookiness of a
complete theory is the apparent ``action at a distance'' due to outcome correlations
\cite{born1971born}.
We show that spookiness for complete theories is equivalent to Schr\"odinger's steering property
\cite{schrodinger1936probability,  hughston1993complete}. We do not discuss the completeness
assumption since an exhaustive inquiry would require a much 
more complicate analysis, comparable to a generalized Bell theorem for GPTs. Finally, we use the
results about spookiness to prove that complementarity and steering are necessary and sufficient
conditions to raise a Schr\"odinger's cat-like paradox.

\section{Hidden variable theories for a GPT}

The most important feature of a given probabilistic theory---such as
QT or more generally any GPT---is the probability rule that links the
various elements of the theory itself. More precisely, given a state
$\rho$, 
a group of observers for the theory ($ A, B, C, \dots $) and the
measurements $ a,b,c,\dots $ that $ A,B,C,\dots $ perform, the
probability rule $ \ProbCond{ a_i, b_j, c_k, \dots }{ a, b, c, \dots
  \rho} $ is defined for every possible outcome $ a_i, b_j, c_k, \dots
$ over a suitable sample space $\Omega$.  In the remainder of the
Letter, we will drop the explicit dependence of all probablity rules
on the state $\rho$. We can now define a hidden variable description
for the previous model as follows.
\begin{Def}[Hidden Variable Theory]
  An equivalent HVT for a GPT is given by a set $\Lambda\ni\lambda$, and a probability rule
  $\ProbCondTilde{\cdot}{\cdot}$ on $ \Omega \times \Lambda$, such that \cite{Branderburger08}
\begin{align}\label{eq:HVT}
  & \ProbCond{ a_i, b_j, c_k, \dots}{ a,b,c,\dots }=\\
  & \sum_{\lambda} \ProbCondTilde{a_i, b_j, c_k, \dots}{
    a,b,c,\dots, \lambda}\ProbCondTilde{\lambda}{ a,b,c, \dots
  }. \nonumber
\end{align}
for every state of the GPT.
\end{Def}
In the following we will restrict our attention to HVTs satisfying two
requirements: {\em $\lambda$-independence}, namely
$\ProbCondTilde{\lambda}{ a,b,c,
  \dots}=\widetilde{\mathrm{Pr}}[\lambda]$, i.e.\ $\lambda$ is an {\em
  objective} parameter independent of the choice of
measurements \footnote{Notice that a realistic theory where $\lambda $ is
  correlated with the observers' choices could in principle be considered,
  however such a theory would be necessarily {\em ad hoc}, and even
  more puzzling than its original GPT \cite {bell}.}; and {\em
  parameter independence}, namely $\ProbCondTilde{a_i}{ a,b,c, \dots,
  \lambda}=\ProbCondTilde{a_i}{a,\lambda}$ and similarly for $b$, $c$,
$\dots$, i.e.\ the HVT is {\em no-signaling}. Clearly, given a GPT,
without these two restrictions we can always build an equivalent
deterministic HVT which is signaling, and denies observers' free
choice \cite{Branderburger08}.

A GPT is {\em complete} if every equivalent HVT provides no further descriptive detail.  Besides
classical probability theory, there is at least a GPT that is complete in the present sense, which
is indeed Quantum Theory, as proved recently by Colbeck and Renner in Ref. \cite{Colbeck2011}.  

It is now crucial to require that probabilities depend non-trivially on the hidden variable.
\begin{Def}[Descriptively significant HVT]
  A HVT is descriptively significant for an equivalent GPT if it satisfies $\lambda$-independence
  and parameter independence, and there exists a pure state and measurements $a,b,\dots, a_i,
  b_j,\dots$ such that for some $\lambda,\lambda^\prime\in\Lambda$ with $\ProbCondTilde{a_i, b_j,
    \dots}{a,b,\dots, \lambda} \ne 0$, one has
  \begin{equation}\label{eq:nonsignificative}
\ProbCondTilde{a_i, b_j, \dots}{a,b,\dots, \lambda} 
\neq\ProbCondTilde{a_i, b_j, \dots}{a,b,\dots, \lambda^\prime}.\!\!\!
  \end{equation}
\end{Def}

\begin{Def}[Complete GPT]
A GPT is complete if every equivalent HVT is not descriptively significant.
\end{Def}

The reason why it is important to investigate only descriptively significant HVTs is the
following. Given a non significant HVT for a given GPT, for all pure states and all $a,b,\dots, a_i,
  b_j,\dots$, we have that, by Eq.\eqref{eq:nonsignificative} and Eq.\eqref{eq:HVT}
\begin{equation}
  \ProbCond{a_i,b_j,\dots}{a,b,\dots} = \ProbCondTilde{a_i,b_j,\dots}{a,b,\dots,\lambda_i},
\end{equation}
for all \( \lambda_i\in\Lambda \) such that $\ProbCondTilde{a_i, b_j,
\dots}{a,b,\dots, \lambda_i} \ne 0$. Therefore, we conclude that \( \ProbCondTilde{a_i,b_j,\dots}{a,b,\dots,\lambda_i} \) shares all the features of \( \ProbCond{a_i,b_j,\dots}{a,b,\dots} \), e.g.\ non locality or complementarity.

Given a GPT, among all HVTs equivalent to it and not descriptively significant, there is one theory
that enjoys the so-called ``\emph{single-valuedness property}'' \cite{Branderburger08}.

\begin{Def}[Single-valuedness]\label{sval}
  A HVT satisfies the single-valuedness property if $ \left| \Lambda \right| = 1 $.
\end{Def}
For a HVT with single-valuedness there exists only one hidden variable value
$\lambda_0$, whence for every $i$ and $j$, $ \ProbCond{a_i, b_j}{a, b}
\equiv \ProbCondTilde{a_i, b_j}{a, b, \lambda_0}$. Given a GPT there
is always an equivalent hidden variable model which satisfies
single-valuedness \cite{Branderburger08}: this fact recalls the
intuition that QT can be regarded itself as a HVT, where the hidden
variable role is played by the quantum state. If we want to study a
complete probabilistic theory it is useful to refer to the simplest
non descriptively significant equivalent hidden variable model, that
is the one which satisfies single-valuedness.

Thanks to J.~P.~Jarrett \cite{Jarrett84}, it is known that the Bell
locality \cite{Bell_1964} is equivalent to the conjunction of two
different properties: the aforementioned Parameter Independence
and the Shimony's so-called Outcome Independence \cite{shimony1984controllable}.
Parameter independence corresponds to the property of  ``\emph{no-signaling without
  exchange of physical systems}'' in \cite{QUIT-ProbTheories} for
GPTs, while Outcome Independence can be stated as the factorizability
of joint probabilities, i.e.\ \footnote{The usual definition of {\em Outcome
  independence} in the literature is the following. A probabilistic HVT
  satisfies the outcome independence property if and only if $ \forall
  a,b,c,\dots,a_i,b_j,c_k,\dots, \lambda $ on \begin{eqnarray*}
  &&\operatorname{Pr}[{a_i}|{a,b,c,\dots,b_j,c_k,\dots, \lambda}]
  =\operatorname{Pr}[{a_i}|{a,b,c,\dots, \lambda}],\\
  &&\operatorname{Pr}[{b_j}|{a,b,c,\dots,b_j,c_k,\dots, \lambda}] =
  \operatorname{Pr}[{b_j}|{a,b,c,\dots, \lambda}],\\
  &&\operatorname{Pr}[{c_k}|{a,b,c,\dots,b_j,c_k,\dots, \lambda}] =
  \operatorname{Pr}[{c_k}|{a,b,c,\dots, \lambda}], \end{eqnarray*} and so on.
  One can easily prove that this definition is equivalent to Eq.
  (\ref{eq:OI-factorisability}).}
\begin{equation}\label{eq:OI-factorisability}
  \ProbCondTilde{a_i, b_j}{a, b, \lambda} = \ProbCondTilde{a_i}{a, b,
\lambda} \times \ProbCondTilde{b_j}{a, b, \lambda}.
\end{equation}
Notice that the previous definition can be applied to a general GPT,
regarded as a single-valued HVT.

The EPR paradox can be rewritten in the following similar way
\cite{EPRnorsen, Branderburger08}: quantum predictions are not
compatible with any equivalent non descriptively significant HVT which
satisfies Outcome Independence. For this reason, according to EPR,
QT presents a {\em spooky action at a distance}. We
now want to extend the EPR result, namely: which are the GPTs that
present this spooky flavor? First we must define in what sense a
theory can present spooky features.
\begin{Def}[Spooky theory]
  A GPT is spooky if it violates outcome independence on a pure state and every equivalent descriptively
  significant HVT does so.
\end{Def}
From now on, we will focus on complete spooky GPTs, unless told
otherwise.

\section{Review of general probabilistic theories}

Before starting we need to introduce the usual notation for GPTs.  For
a detailed discussion see \cite{QUIT-Arxiv}. The symbols $\rho_A$,
$\State{\rho}_A $ and $ \Qcircuit @C=.5em @R=.5em { \prepareC{\rho} &
\ustickcool{A} \qw & \qw } $ denote the \emph{state} $\rho$ for system
$A$, representing the information about the system initialization,
including the probability that such preparation can occur. The set of
the states of a given system $A$ is a (truncated) positive cone, and
therefore given the states $\{ \rho_{i} \}_{i\in\eta}$ for $A$, every their
convex combination belongs to the cone of the states of $A$. The extremal rays of the
cone---namely the states which cannot  be seen as a convex combination
of other ones---are the so called {\em pure} states.

Similarly, $c_{iA}$, $\Effect{ c_i }_A $ and $ \Qcircuit @C=.5em
@R=.5em { & \ustickcool{A} \qw & \measureD{ c_i } } $ mean the
\emph{effect} $c_i$ for system $A$ or, in more practical terms, the
$i$-th outcome of the test (measurement) $c=\{c_i\}_{i\in\eta}$ on
system $A$. Given a system $A$, its effects are bounded linear
positive functionals from the states of $A$ to $\left[ 0,1
\right]\subset\mathbb{R}$, and therefore they belong to the dual cone
of the cone of the states. The application of the effect $c_i$ on the
state $\rho$ is written as $\RBraKet{ c_i }{ \rho }_A$ or $\Qcircuit
@C=.5em @R=.5em { \prepareC{\rho} & \ustickcool{A} \qw & \measureD{
c_i } } $, and it means the probability that the outcome of measure
$c$ performed on the state $\rho$ of system $A$ is $c_i$, i.e.\ $
\RBraKet{ c_i }{ \rho }_A :=\ProbCond{c_i}{c}$. In the following we
will not specify the system when it is clear from the context or it is
generic.

The symbol $e_A$ will denote a \emph{deterministic effect} for
system $A$, namely a measurement with a single outcome.  For any state
$\sigma$, the symbol $\RBraKet{ e }{ \sigma }$ denotes its preparation
probability within a test including a measurement $\{c_i\}_{i\in\eta}$
such that $e=\sum_{i\in\eta} c_i$. A state $\sigma$ is deterministic
if we know with certainty that it has been prepared in any test,
whence $ \RBraKet{ e }{ \sigma }= 1 $ for every deterministic effect
$e$. An {\em ensemble} is a collection of (possibly non-deterministic)
states $\{\alpha_i\}_{i\in\eta}$ such that
$\rho:=\sum_{i\in\eta}\alpha_i$ is deterministic.  A GPT is causal
(i.e.\ it satisfies the {\em no-signaling from the future} axiom
\cite{QUIT-Arxiv}) iff the deterministic effect is unique. Thanks to
this last feature, in a causal GPT the preparation probability for the
state $\sigma$ is well defined since it is independent of the tests
following the preparation. For this reason, for every state $\rho$ we
can always consider the deterministic state $\bar\rho :=
\RBraKet{e}{\rho}\rho$, or, in other words, in a causal theory
evey state is proportional to a deterministic one. In the following,
we will consider a general {\em causal} GPT.

\section{Spookiness, steering and completeness}

In this section we will show our main results. Let $\rho$ be a joint
deterministic state for systems $A$ and $B$. Let $a_0$, $a_1$ be two
effects for $A$ forming a so-called \emph{complete test}: namely, for
every state $ \sigma$ of $A$ we have $ \RBraKet{ a_0 }{ \sigma}+
\RBraKet{a_1}{ \sigma }= \RBraKet{e}{ \sigma} $. Similarly, let the
effects $b_0$, $b_1$ form a complete test for $B$. Let us define the
following useful shorthand
\begin{equation}\label{eq:abbreviations}
 p_{ij} := \ProbCond{a_i, b_j}{a, b} \equiv 	\begin{aligned}
												  \Qcircuit @C=.5em @R=.5em { \multiprepareC{1}{ \rho }& \ustickcool{A} \qw & \measureD{a_i} \\
																			   \pureghost{ \rho } &  \ustickcool{B} \qw &  \measureD{b_j} }
												\end{aligned}.
\end{equation}
The number $ p_{ ij } $ represents the probability that Alice and Bob
obtain respectively the $i$-th and the $j$-th outcome while performing
measurements $a$ and $b$ on the state $\rho$. 

Under these assumptions the following theorems hold.

\begin{theorem}\label{thm:NotOI=Table}
  A complete GPT is spooky if and only if there exists a pure state $\rho$ of $AB$ and tests $\{a_0,a_1\}$ of $A$ and
  $\{b_0, b_1\}$ of $B$ such that the probabilities $p_{ij} $ of Eq.~\eqref{eq:abbreviations}
  satisfy the following constraint: 
  \begin{equation}\label{eq:spooky}
    p_{00} \ p_{11}\ne p_{01}p_{10} .
  \end{equation}
\end{theorem}

\begin{proof}
  A complete GPT is spooky iff it has a test $\left\{ a_0,a_1 \right\}$ for system $A$, a
  test $\{b_0,b_1\}$ for system $B$, and pure state $\rho$ for $AB$ such that $p_{ij}$ is not
  factorized. This is equivalent to the requirement that the matrix $p_{ij}$ has rank larger than
  one, namely for a $2\times 2$ matrix the determinant of the matrix is non vanishing, i.e. Eq.
  (\ref{eq:spooky}) holds.
\end{proof}

\begin{figure}
  \includegraphics[width=.85\columnwidth]{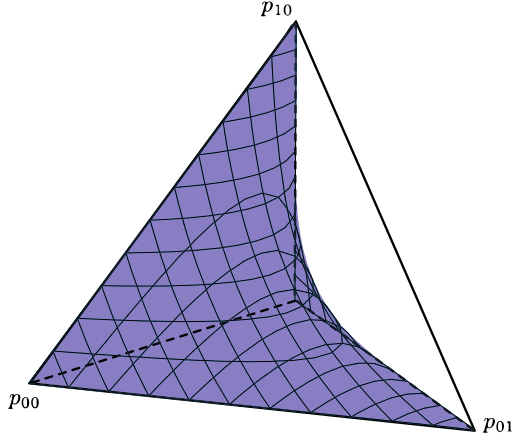}
  \caption{\label{fig:tetrahedron_factorisability_surface} The tetrahedron (outlined in thick black
    line) represents the set of possible values of probabilities $p_{00}$, $p_{01}$, $p_{10}$
    satisfying the normalization condition. The hyperbolic paraboloid identifies the probabilities
    that satisfy Outcome Independence.} 
\end{figure}

The theorem states that if we describe by a pure state of a GPT an experiment with probabilities
given by Eq.\eqref{eq:spooky}, then we can not provide a local explanation for the observed physical
phenomenon.

Notice that the previous result does not manifestly require any choice by the observers, since it needs only
one test for each subsystem.  Consequently, there are no explicit assumptions about the observers' free will,
thus extending Brandenburger--Yanofsky's reformulation of the
EPR paradox \cite{Branderburger08}. One may wonder how such non-locality can be proved without 
observers' choice between ``complementary'' measurements. The answer resides in the completeness
requirement. Clearly a GPT with only one measurement for each observer always admits an equivalent HVT,
which is the deterministic one, since there is no requirement for parameter independence and
$\lambda$-independence. 

Generally, requiring $ 0 \le p_{ij}\le 1 $ and $ \sum p_{ij} = 1 $ implies
that $ p_{00} $, $ p_{01} $, $ p_{10} $ must lie in the tetrahedron outlined in figure
\ref{fig:tetrahedron_factorisability_surface} ($p_{11}$ is simply obtained by the normalization
condition). The spookiness condition, namely Eq.~\eqref{eq:spooky}, defines a hyperbolic paraboloid, and
all spooky theories give rise to points of the tetrahedron that do not lie on the surface of the
paraboloid.

We now prove a theorem that along with Theorem \ref{thm:NotOI=Table}
provides the main result in this Letter. The following theorem
pertains to the property of steering \cite{schrodinger1936probability,
  hughston1993complete} for a GPT, that we briefly recall here.
\begin{Def}[Steering state for an ensemble]
  The state \( \rho \) of the system $AB$ steers the ensemble \( \{
  p_i\alpha_i \}_{i\in\eta} \) of states of the system $A$ if there
  exists a test \( \{ b_i \}_{i\in\eta} \) of the system $B$ such that
\begin{equation}
  \begin{aligned}
    \Qcircuit @C=.5em @R=.5em { \multiprepareC{1}{ \rho }& \ustickcool{A} \qw & \qw \\
							    \pureghost{ \rho } & \ustickcool{B} \qw &  \measureD{b_i} }
  \end{aligned} \equiv \ p_i\ \Qcircuit @C=.5em @R=.5em { \prepareC{\alpha_i} & \ustickcool{A} \qw & \qw } \qquad \left( \forall i\in\eta \right).
\end{equation}
\end{Def}

\begin{theorem}\label{thm:Table=Steering}
  A GPT admits a steering state $\rho$ for a non trivial ensemble of
  two different states if and only if the probabilities of $\rho$ satisfy
  Eq.~\eqref{eq:spooky} for some local test $a,b$.
\end{theorem}
\begin{proof} Let us prove the two-way implication in two steps.
  \paragraph{$\left( \Rightarrow \right) $} The steering assumption
  implies the existence of a state $\rho$ for the composite system $
  AB $ which steers the marginal ensemble $\{p_0 \alpha_0,
  p_1\alpha_1\}$, with $ 0 < p_0, p_1 < 1$, $ p_0 + p_1 = 1 $ and
  $\alpha_i$ deterministic states for $A$, such that $\RBraKet{ a_0
  }{\alpha_0 } \ne \RBraKet{ a_0 }{ \alpha_1 } $ for some effect
  $a_0$.
This last inequality implies that there
exists $ 0
< \left| w \right| \le 1 $ such that $ \RBraKet{ a_0 }{ \alpha_0 } =
\RBraKet{ a_0 }{ \alpha_1 } + w $ (or equivalently $ \RBraKet{ a_1
}{ \alpha_0 } = \RBraKet{ a_1 }{ \alpha_1 } - w $). Therefore,
using the substitutions of Eq.~\eqref{eq:abbreviations},
the RHS of Eq.~\eqref{eq:spooky} $ p_{01}p_{10}= p_0 p_1
\RBraKet{ a_0 }{ \alpha_1 } \RBraKet{ a_1 }{ \alpha_0 }$ can be
rewritten as
\begin{align}
  &p_0 p_1 \{\RBraKet{ a_0 }{ \alpha_0 }\RBraKet{ a_1 }{ \alpha_1 }-w[\RBraKet{ a_0 }{ \alpha_0 }+\RBraKet{ a_1 }{ \alpha_1 }] +w^2\}\nonumber\\
  &=p_0 p_1 \{\RBraKet{ a_0 }{ \alpha_0 }\RBraKet{ a_1 }{ \alpha_1 }-w\},
\end{align}
where we used $ \RBraKet{ a_1 }{ \alpha_1 } =\RBraKet{ a_1 }{ \alpha_0 } + w$ and the
normalization $\RBraKet{ a_0 }{ \alpha_0 }+\RBraKet{ a_1 }{ \alpha_0 }=1$.
Since $ w \ne 0 $ and $ 0 < p_0, p_1 < 1$, we conclude that the RHS of
Eq.~\eqref{eq:spooky} is not equal to $ p_0 p_1 \RBraKet{ a_0 }{ \alpha_0 }\RBraKet{ a_1}{ \alpha_1 }
= p_{00}p_{11} $, thus proving Eq.~\eqref{eq:spooky}.

\paragraph{$\left( \Leftarrow \right) $} Let us introduce for
system $A$ the (non-deterministic) states $\tilde{\alpha}_0 $, $
\tilde{\alpha}_1$ defined as
\begin{equation}\label{eq:AlphaTilde}
  \begin{aligned}
    \Qcircuit @C=.5em @R=.5em { \prepareC{\tilde\alpha_i} &
      \ustickcool{A} \qw & \qw }
  \end{aligned} :=
  \begin{aligned}
    \Qcircuit @C=.5em @R=.5em { \multiprepareC{1}{\rho} 	& \ustickcool{A} \qw & \qw \\
      \pureghost{\rho} & \ustickcool{B} \qw & \measureD{b_i} }
  \end{aligned} \qquad \left( i = 0,1 \right).
\end{equation}
Thanks to Eqs.~(\ref{eq:abbreviations},\ref{eq:AlphaTilde}),
Eq.~\eqref{eq:spooky} can be rewritten as follows
\begin{align}
  \RBraKet{ a_0 }{ \tilde\alpha_0 }\RBraKet{ a_1 }{ \tilde\alpha_1} = \RBraKet{ a_0 }{ \tilde\alpha_1 }\RBraKet{ a_1 }{ \tilde\alpha_0 } + w,
\label{eq:intra}
\end{align}
where $ 0 < \left| w \right| \le 1 $. It is useful to define the deterministic states $ \alpha_0$, $
\alpha_1$ for system $A$ such that
\begin{equation}
  \label{eq:Alpha} \tilde\alpha_i = \RBraKet{e}{\tilde\alpha_i} \alpha_i \qquad \left( i = 0,1 \right).
\end{equation}
Since $ \{a_0, a_1\} $ is a complete test, $
\RBraKet{ a_0 }{\tilde \alpha_i } + \RBraKet{ a_1 }{\tilde \alpha_i
} =\RBraKet{e}{\tilde\alpha_i}$, for $i=0,1$, and from
Eq.~\eqref{eq:intra} we conclude that $ \RBraKet{ e }{ \tilde\alpha_0
} $ and $ \RBraKet{ e }{ \tilde\alpha_1 } $ cannot be zero
(otherwise $w$ would be zero, against the hypothesis).  Thus
Eq.~\eqref{eq:intra} can be divided by $ \RBraKet{ e }{ \tilde\alpha_0
} \RBraKet{ e }{ \tilde\alpha_1 }$, obtaining
\begin{equation}
  \RBraKet{ a_0 }{\alpha_0 }\RBraKet{ a_1 }{ \alpha_1} = \RBraKet{ a_0 }{\alpha_1 }\RBraKet{ a_1 }{ \alpha_0 } +\frac{w}{\RBraKet{ e }{ \tilde\alpha_0 } \RBraKet{ e }{ \tilde\alpha_1 } }.
\end{equation}
Using $\RBraKet{ a_1 }{\alpha_i }=1-\RBraKet{ a_0 }{\alpha_i}$, one has
\begin{equation}
  \RBraKet{ a_0 }{ \alpha_0 } =  \RBraKet{ a_0 }{ \alpha_1 } + \frac{w}{\RBraKet{ e }{ \tilde\alpha_0 } \RBraKet{ e }{ \tilde\alpha_1 } }.
\end{equation}
The last term of the RHS of the previous equation is not equal to
zero, since $ w \ne 0 $,  thus we conclude that $ \RBraKet{ a_0 }{
  \alpha_0 } \ne \RBraKet{ a_0 }{ \alpha_1 } $.
Since $ \RBraKet{e}{\tilde\alpha_0}, \RBraKet{e}{\tilde\alpha_1} \ne 0, 1 $ the ensemble
$\{\RBraKet{e}{\tilde\alpha_0}\alpha_0,\RBraKet{e}{\tilde\alpha_1} \alpha_1\}$
is not trivial. Finally, according to \eqref{eq:AlphaTilde}, and remembering that $b_0$, $b_1$
constitute a complete test for system $B$, it can be easily seen that 
the state $\rho$ for the composite system $AB$ steers the non trivial ensemble
$\{\RBraKet{e}{\tilde\alpha_0} \alpha_0,\RBraKet{e}{\tilde\alpha_1} \alpha_1 \}$
thanks to effects $b_0$, $b_1 $.
\end{proof}
As a natural consequence of Theorems \ref{thm:NotOI=Table} and \ref{thm:Table=Steering}, we have the following corollary.
\begin{corollary}
For a complete GPT the conditions: i) spookiness, ii) existence of a pure steering state for a
non-trivial ensemble, and iii) existence of a pure state satisfying Eq.~\eqref{eq:spooky}, are all
equivalent.
\end{corollary}
Notice that completeness is always assumed in our arguments (apart from Theorem
\ref{thm:Table=Steering}). Indeed, since the hyperbolic paraboloid of Fig.
\ref{fig:tetrahedron_factorisability_surface} includes the four vertices of the tetrahedron, the
probabilities of a single couple of tests $a=\{a_0,a_1\}$ for $A$ and $b=\{b_0,b_1\}$ for $B$ can
always be thought of as a mixture of factorized probabilities, and consequently there could be in
principle a descriptively significant HVT compatible with any such model. In order to state stronger
no-go theorems---like Bell's inequality---one must consider incompatible tests, and the assumption
of free will becomes crucial.

\section{Complementarity and Schr\"odinger's cat}

From the results in the previous section the following corollary can be
easily proved.
\begin{corollary}
  A complete GPT with a pure steering state for a mixture of two
  states $\alpha_0,\alpha_1$ with $\alpha_1$ conclusively
  discriminable from $\alpha_0$ is spooky.
\end{corollary}
 
Before proving the last corollary, we precisely define when two
states are probabilistically discriminable.

\begin{Def}[Conclusively discriminable states]
  The state $ \alpha_1 $ is conclusively discriminable from the state
  $ \alpha_0 $ if there exists an effect $ a $ such that $
  \RBraKet{a}{\alpha_0} = 0 $ and $ 0 < \RBraKet{a}{\alpha_1} \le 1 $.
  If $ \RBraKet{a}{\alpha_1} = 1 $ we say that $\alpha_0 $, $\alpha_1
  $ are perfectly discriminable.
\end{Def}

\begin{proof}[Proof of the corollary]
  Two states $\alpha_0,\alpha_1$ with $\alpha_1$ conclusively
  discriminable from $\alpha_0$ provide a particular case of different
  states. Therefore we can apply Theorem \ref{thm:Table=Steering} and
  conclude that the probabilities for the GPT must reside in the
  tetrahedron and not on the hyperbolic paraboloid. Hence, according
  to Theorem \ref{thm:NotOI=Table}, the complete GPT is spooky.
\end{proof}
We will now show that complementarity--along with steering--implies all variants of the
Schr\"odinger-cat paradox. The notion of complementarity has been the main focus of Bohr's
philosophy of QT, however, it has been often criticized for the lack of a precise mathematical
formulation. A definition of complementarity is provided in the framework of quantum logic (see
e.~g. \cite{Lahti:1980p10323} and references therein), however, it has never been defined as a
general notion outside QT.  This is due to the fact that complementarity regards contexts that may
seem unrelated, as wave-particle duality and non-commutativity.  Uncertainty and its quantitative
relation with non-locality was analyzed in Ref.  \cite{Oppenheim19112010}. Here we propose a notion
of complementarity that summarizes all the aspects that emerge within QT, and allows for a precise
mathematical formulation within the broader context of GPTs.  In order to do that, let us define
what is a proposition for a GPT.

\begin{Def}[Proposition for a GPT]
  Given a GPT, let $ a:=\{a_0,a_1\}$ be a complete binary test. The test $a$ is a {\em proposition}
  if there exist two states $\alpha_0$ and $\alpha_1$ such that $\RBraKet{ a_i }{ \alpha_j }=
  \delta_{ij} $.
\end{Def}

We will call a state $\rho$ {\em sharp for a set of propositions $\{a^{(i)}\}$} if the 
probabilities for all effects of such propositions are either zero or one. We can now precisely
formulate complementarity. 

\begin{Def}[GPT with Complementarity]\label{def:GPT-complementarity}
  A GPT entails complementarity if there are two propositions $a^{(0)}$ and $a^{(1)}$ having no
  common sharp state.  These propositions will be called {\em complementary}.
\end{Def}
One may think that a more general definition of complementarity involves a number $N\geq 2$ of
propositions $\{a^{(i)} \}_{i=0}^N$ having no common sharp state. However, this is just
equivalent to Definition \ref{def:GPT-complementarity}, namely complementarity is an intrinsically
dual notion \footnote{Indeed, by hypothesis there exist $N+1$ tests $\{a^{(i)}\}_{i=0}^N$ such that 
\(
\RBraKet{a^{(i)}_j}{\rho} = 1 \implies \exists k,l\quad 0 < \RBraKet{a^{(k)}_l}{\rho} < 1
\).
Let us define the number $k$ as the maximum number of the propositions \( a^{(i)} \) for which there exists a state $\rho$ such that each proposition is deterministic, i.e.
\begin{align*}
  &\Phi:=\{\phi\subseteq\{0,\dots,N\}|\exists\rho,\ \forall j\in\phi\ \ \exists l:\ \RBraKet{a^{(j)}_l}{\rho}=1\},\\
  &k :=\max_\Phi |\phi|.
\end{align*}
By hypothesis, the number $k$ is strictly less than \( N+1 \). Let us take a set $\phi\in\Phi$ for which $|\phi|=k$, and let us define the effects \( a_j := \frac{1}{k}\sum_{i\in\phi} a^{(i)}_j\) and \( \tilde a_j := a^{(l)}_j \) with an arbitrary \( l\notin\phi \). By construction, both tests \( a:= \{ a_j \} \) and \( \tilde a := \{ \tilde a_j \} \) are propositions and have no common sharp state. The converse is trivial.}.

By the above definition the complementary propositions cannot jointly have a definite truth value.
What is the paradox of the famous Schr\"odinger's cat argument \cite{schrodinger1935present}? In its
popularized version the paradox lies in the fact that the cat pure state is a superposition
$\Ket{\Psi^\pm}:=2^{-\frac{1}{2}}(\Ket{\text{dead} }\pm\Ket{\text{alive}})$ before the measurement of its state of
life, thus coming from complementarity \emph{per se} (the test
$\{\KetBra{\Psi^+}{\Psi^+},\KetBra{\Psi^-}{\Psi^-}\}$ is complementary to $\{\KetBra{ \text{dead}
}{\text{dead}}, \KetBra{ \text{alive} }{\text{alive}}\} $.  However, the original paradox is subtler
and relies on the ability to remotely prepare orthogonal states for the cat. Let us imagine for
example that the state of life of the cat is entangled with the spin of an electron as in the state
$2^{-\frac{1}{2}}(\Ket{ \uparrow } \otimes \Ket{ \text{alive} } +\Ket{ \downarrow } \otimes \Ket{
  \text{dead} }) $. After the measurement of the spin of the electron along the $z$ direction we
have prepared the cat in the states $ \Ket{ \text{dead}} $ or $ \Ket{\text{alive}} $ each with
probability $1/2$. The proposition corresponding to the life state of the cat has a truth value that
is conditioned by the outcome of a measurement on the electron. This situation by itself would not
be puzzling if the state were a mixture, as in the Bell's argument of {\em Bertlmann's socks}.  The
paradox is the fact that in a {\em pure state} a definite property--the cat is alive--is neither
true nor false. This version of the paradox stems from pure state steering of an ensemble of
perfectly discriminable states.

We now provide a third version of the paradox, which relies on the existence of a pure state that
steers an ensemble of sharp states for complementary propositions. This is the case e.g. of the
state $2^{-\frac{1}{2}}(\Ket{ \uparrow } \otimes \Ket{ \text{alive} } +\Ket{ \downarrow }
\otimes\Ket{\Psi^+}) $. After the measurement of the spin of the electron along the $z$ direction we
have prepared the cat in the states $ \Ket{ \text{alive}} $ or
$2^{-\frac{1}{2}}(\Ket{\text{dead}}+\Ket{\text{alive}}) $ each with probability $1/2$.  In this case
the outcome of the measurement does not simply decide the truth value of a proposition, but it even
establishes which proposition has a definite truth value.  It is worth noticing that, according to
Definition \ref{def:GPT-complementarity}, the tests $\{ \KetBra{ \text{alive} }{ \text{alive}
},\KetBra{ \text{dead} }{ \text{dead}}\}$ and $\{\KetBra{ \Psi^+ }{ \Psi^+ } ,\KetBra{ \Psi^- }{
  \Psi^- }\}$ are complementary.  Complementarity and steering are thus the ingredients for the
third version of the paradox: given a GPT, suppose that for a system $A$ (the cat) there are two
complementary propositions $ a $, $ \tilde a $. By hypothesis there are two sets of states $ \{
\alpha_i \} $, $ \{ \tilde\alpha_i \} $ such that $ \RBraKet{ a_i }{ \alpha_j } = \RBraKet{ \tilde
  a_i }{ \tilde\alpha_j } = \delta_{ij} $, $ 0 < \RBraKet{ a_i }{ \tilde\alpha_j }, \RBraKet{ \tilde
  a_i }{ \alpha_j }< 1 $. If the theory has a pure steering state $ \rho_{AB} $ for the ensemble
$\{p_0 \alpha_0, p_1 \tilde \alpha_0 \} $ of the system $A$, thanks to our corollary we conclude
that the GPT is spooky, since $ \tilde \alpha_0 $ (``alive'') is conclusively discriminable from $
\alpha_0$ ($\Psi^+$) by test $ a$. We notice that the first version of the paradox involves only
complementary, the second one involves only pure-state steering, whereas the third one uses both.
If the theory is complete, the existence of a pure steering state for a perfectly
discriminable or complementary ensemble implies spookiness, which is thus necessary for the second
and third version of the paradox.

\section{Conclusion}

We have shown that for a complete GPT spookiness and pure state steering of a non-trivial ensemble
are equivalent. Moreover, we thoroughly introduced the notion of complementarity for GPTs, and used
it in order to discuss three different versions the Schr\"odinger cat paradox. A crucial ingredient
for all our results is completeness, namely the property of a GPT consisting in the impossibility of
having descriptively significant HVTs. Classical probability theory is complete, and the same has
been recently proved also for QT \cite{Colbeck2011}. In our knowledge QT is the only theory
satisfying the hypotheses of our theorems.  Nevertheless, our result is relevant, due to its
generality, and because it highlights the interplay between two main features of the
theory--spookiness and existence of a pure steering state--without recurring to the mathematical
structure of Hilbert spaces, only relying on the conceptual formalism of GPTs.  The question whether
the theorem applies to a wider class of theories opens a decisive new problem, namely determining
what GPTs are complete, and--if other than Classical and Quantum--what are the common features they
enjoy.

\section*{Acknowledgments}
  We acknowledge useful comments by H. Wiseman. This work is supported by Italian
  Ministry of Education through grant PRIN 2008. P. P. acknowledges financial support by the EU
  through FP7 STREP project COQUIT.

\bibliography{spooky}
\end{document}